\documentclass[aps,pra,twocolumn,superscriptaddress,floatfix,
nofootinbib,showpacs,longbibliography]{revtex4-2}

\usepackage[normalem]{ulem}
\usepackage{braket}
\usepackage{xcolor}
\usepackage[utf8]{inputenc}  
\usepackage[T1]{fontenc}     
\usepackage[british]{babel}  
\usepackage[sc,osf]{mathpazo}
\usepackage[scaled=0.86]{berasans}  
\usepackage[colorlinks=true, citecolor=blue, urlcolor=blue]{hyperref}  
\usepackage{graphicx} 
\usepackage[babel]{microtype}  
\usepackage{amsmath,amssymb,amsthm,bm,amsfonts,mathrsfs,bbm} 

\usepackage{xspace}  
\usepackage{pgf,tikz,pgfplots}
\usetikzlibrary{calc}
\pgfplotsset{compat=1.15}
\usepackage{mathrsfs}
\usetikzlibrary{arrows,snakes}
\usetikzlibrary{positioning}
\usepackage{xcolor}
\usepackage{appendix}
\usepackage{multirow}
\usepackage{array}
\usepackage{bigstrut}
\usepackage{braket}
\usepackage{color}
\usepackage{natbib}
\usepackage{multirow}
\usepackage{float}
\usepackage[caption = false]{subfig}
\usepackage{xcolor,colortbl}
\usepackage{color}

\newtheorem{theorem}{Theorem}

\newtheorem{proposition}{Proposition}

\newcommand{\tri}[3]{\ket{#1}\ket{#2}\ket{#3}}
\begin{document}

\title{Activating Strongest Possible Nonlocality from Local Sets: An Elimination Paradigm}
	
\author{Subhendu B. Ghosh}
\affiliation{Physics and Applied Mathematics Unit, Indian Statistical Institute, Kolkata, 203 B. T. Road, Kolkata 700108, India}

\author{Tathagata Gupta}
\affiliation{Physics and Applied Mathematics Unit, Indian Statistical Institute, Kolkata, 203 B. T. Road, Kolkata 700108, India}

\author{Ardra A. V.}
\affiliation{School of Physics, IISER Thiruvananthapuram, Vithura, Kerala 695551, India}

\author{Anandamay Das Bhowmik}
\affiliation{Physics and Applied Mathematics Unit, Indian Statistical Institute, Kolkata, 203 B. T. Road, Kolkata 700108, India}

\author{Sutapa Saha}
\affiliation{Physics and Applied Mathematics Unit, Indian Statistical Institute, Kolkata, 203 B. T. Road, Kolkata 700108, India}

\author{Tamal Guha}
\affiliation{Department of Computer Science, The University of Hong Kong, Pokfulam Road, Hong Kong}

\author{Amit Mukherjee}
\affiliation{S. N. Bose National Center for Basic Sciences, Block JD, Sector III, Saltlake, Kolkata 700098, India}
\begin{abstract}
Apart from the Bell nonlocality, which deals with the correlations generated from the local input-output statistics, quantum theory exhibits another kind of nonlocality that involves the indistiguishability of the locally preparable set of multipartite states. While Bell-type nonlocality cannot be distilled from a given local correlation, it is already reported that the latter kind of nonlocality can be activated from a "local", i.e., locally distinguishable set of states. Although, recently it is shown that a stronger notion of such a nonlocality, which deals with elimination instead of discrimination, can be activated from locally preparable bipartite states of dimension $7\times8$, the present work observes that the same notion can be demonstrated even in lower dimensional multipartite systems. Importantly, the strongest possible version of such an activation is further depicted here, where none of the transformed product states can be eliminated, even if all but one of the parties come together.
\end{abstract}
\maketitle
\section{Introduction}
Nonlocality is the most curious feature of the quantum world. A celebrated nonlocal phenomenon - Bell nonlocality \cite{Bell66} arises in quantum physics due to entanglement. This phenomenon actually does not allow any local realistic explanation. Apart from its fundamental implications, Bell nonlocality is also immensely useful in some intriguing applications \cite{testdim,random,btsqkd,nsqkd,renner12,game2,game1}. However, besides Bell nonlocality, quantum theory admits a different kind of multipartite nonclassicality, which identifies a set of orthogonal states which are impossible to distinguish under local operations and classical communications (LOCC) 
\cite{lid1,lid2,lid3,lid4,lid5,lid6,lid7,lid8,lid9,lid10,lid11,lid12,lid13,lid14,lid15,lid16,lid17,lid18,lid19,lid20,lid21,lid22,lid23,lid24,lid25,lid26,lid28,lid29,lid30,lid31,lid32,lid33,lid34,lid35,lid36,lid37,lid38,lid39}. 
Consequently, a stronger version of state distinguishability has recently been introduced in \cite{halderprl}, where the task is to eliminate some elements from a set of orthogonal states (preserving their orthogonality), shared between the distant parties, under LOCC and such an impossibility is termed as \textit{local irreducibility}. For example, four two-qubit Bell states or, complete three-qubit GHZ basis are locally irreducible \cite{halderprl}. Note that, any locally irreducible set of quantum states are surely locally indistinguishable, while the converse is not true.

Hitherto, we have discussed nonlocal features that are predominantly manifested only by entangled states. Surprisingly, entanglement is not always necessary to exhibit those nonlocal characteristics. Bennett \textit{et al.}, in their pioneering work \cite{bennps}, showed that there exist sets of orthogonal product states which are not perfectly locally distinguishable. This intriguing non-classical feature is termed as `quantum nonlocality without entanglement'. More interestingly, it has also been shown that the set of locally indistinguishable product basis in $\mathbb{C}^3\otimes \mathbb{C}^3$ \cite{bennps} is also locally irreducible \cite{halderprl}.  

However, for multipartite scenarios, manifestation of nonlocality gets more complex. For example, in case of locally irreducible three qubit GHZ basis, it can be shown that whenever any two of the players come together (that is, allowed to perform joint or global measurement) they are able to eliminate some states from the set \cite{halderprl}. The same question of collaboration between all but one party can also be considered in case of multipartite quantum nonlocality without entanglement phenomenon. It can be shown that there is a set of orthogonal tripartite product states in $\mathbb{C}^2\otimes \mathbb{C}^2\otimes \mathbb{C}^2$ \cite{bennps} which are locally indistinguishable when all the players are spatially separated. But if any two of the players are allowed to come together to perform some joint measurement on their composite subsystems, the set will turn out to be a locally distinguishable one. However, \cite{halderprl}, Halder \textit{et al.} have demonstrated strongest form of quantum nonlocality without entanglement in multipartite scenario by introducing a set of tripartite product states that are locally irreducible even when all but one players assemble together. Naturally, this set of states is locally indistinguishable in bipartition. This phenomena is termed as `strong quantum nonlocality without entanglement'. Moreover, a characterization of sets of multipartite nonlocal product states has also been introduced recently \cite{gnps}. These recent works of quantum nonlocality without entanglement in multipartite systems have been fueling a plethora of interesting studies \cite{sristypra,halderpra,senguptapra,somprr,halderpra2,minimal,wupra,shipra,dongpra,zhangpra,yuanpra}.

Apart from revealing the elegant intricacies of state space structure, local indistinguishability and irreducibility of quantum states also indicate the prospect of locking of information such that unlocking requires entangled resources. This characteristic certainly has a crucial significance in various quantum cryptographic schemes \cite{terprl,DiViieee,wernerprl, wehner,markhampra}. It is, therefore, an intriguing question to transform a locally distinguishable set of multipartite states to a set with such nonlocal features via local operations and classical communications. 

In a recent work, Bandyopadhyay and Halder \cite{sar} showed that there exist some locally distinguishable sets of multipartite orthogonal states which can be deterministically transformed to locally indistinguishable sets via orthogonality preserving measurements (OPM) on the local sites of the players. They called this phenomenon genuine activation of nonlocality. The term `genuine' denotes that the set must be free from \textit{local redundancy}.\footnote{{Presence of local redundancy in a set of orthogonal states actually signifies that discarding of one or more subsystems from initial system keeps the orthogonality of the set of reduced states intact. We are, thus, interested in such sets which will be nonorthogonal after discarding one or more subsystems \cite{sar}.}} Apparently, it may seem that the phenomenon of genuine activation is somewhat parallel to the distillation of Bell nonlocality\cite{distprl,distashu}. However, there is a fundamental difference. In case of genuine activation of nonlocality, the initial set does not need to show any nonlocal feature whereas for distillation of Bell-type nonlocality the initial correlation must not be local-realistic. Indeed, in a recent result, Li and Zheng have further extended the search by exploring genuine activation of local irreducibility for bipartite product states \cite{liz} and come up with a locally distinguishable set of bipartite product states which via local OPM can be deterministically converted to a locally irreducible set of bipartite product states. However, the question of genuine activation of nonlocality without entanglement in case of multipartite scenario has yet not been explored in its entirety. In this article, we show that a stronger form of nonlocality of a locally distinguishable set of tripartite orthogonal product states can be genuinely activated. More precisely, the set can be deterministically transformed via local OPMs to a set of orthogonal product states which is locally irreducible even when any two of the three parties come together. Further, we come up with a stronger example  using a set of locally distinguishable states in $\mathbb{C}^{3}\otimes\mathbb{C}^{3}\otimes\mathbb{C}^{6}$, where the measurement performed on the third party's possession can activate the strongest possible  genuine nonlocality without entanglement.

In the following sections, we discuss genuine activation of multipartite quantum nonlocality without entanglement case by case. 

\section{Genuine Activation of Multipartite Quantum Nonlocality Without Entanglement}
We will start our exploration of genuinely activable sets by showing that a locally distinguishable set of bipartite product states can be transformed to a set of locally indistinguishable set of orthogonal states via local orthogonality preserving measurements. 

Consider the set $\mathcal{G}_1\equiv\{\ket{\psi_i}_{AB}\}_{i=1}^5 (\subset \mathbb{C}^3\otimes \mathbb{C}^6)$, where, 
\begin{subequations}
\begin{eqnarray}
\ket{\psi_1}_{AB}&=&\ket{0}_{A}\ket{\mathbf{0}-\mathbf{1}+\mathbf{4}-\mathbf{5}}_{B}\\
\ket{\psi_2}_{AB}&=&\ket{2}_{A}\ket{\mathbf{1}-\mathbf{2}+\mathbf{5}-\mathbf{3}}_{B}\\
\ket{\psi_3}_{AB}&=&\ket{1-2}_{A}\ket{\mathbf{0}-\mathbf{4}}_{B}\\
\ket{\psi_4}_{AB}&=&\ket{0-1}_{A}\ket{\mathbf{2}-\mathbf{3}}_{B}\\
\ket{\psi_5}_{AB}&=&\ket{0+1+2}_{A}\ket{\mathbf{0}+\mathbf{1}+\mathbf{2}+\mathbf{3}+\mathbf{4}+\mathbf{5}}_{B}.
\end{eqnarray}
\end{subequations}
\begin{proposition}
The set $\mathcal{G}_{1}$ is locally distinguishable and free from local redundancy.
\end{proposition}
\begin{proof}
It is quite straightforward to prove that the set $\mathcal{G}_1$ is without local redundancy. Here, Bob's system can be considered to be the composition of qubit and qutrit subsystems. Precisely, $\ket{\mathbf{0}}_B:=\ket{00}_{b_1b_2},\ket{\mathbf{1}}_B:=\ket{01}_{b_1b_2},\ket{\mathbf{2}}_B:=\ket{02}_{b_1b_2},\ket{\mathbf{3}}_B:=\ket{10}_{b_1b_2},\ket{\mathbf{4}_B}:=\ket{11}_{b_1b_2},\ket{\mathbf{5}}_B:=\ket{12}_{b_1b_2}$. Take two states, $\ket{\psi_3}_{AB}$ and $\ket{\psi_4}_{AB}$. When any of the sub-parts (qubit or qutrit) of Bob's system for both states is discarded the reduced states will be non-orthogonal.

Furthermore, $\mathcal{G}_1$ can also be shown to be locally distinguishable.
The players pursue the  following distinguishability protocol. First Bob performs a measurement $N_B\equiv\{N_1:=P[\ket{\mathbf{0}-\mathbf{4}}_B],N_2:=P[\ket{\mathbf{2}-\mathbf{3}}_B],N_3:=P[\ket{\mathbf{0}+\mathbf{1}+\mathbf{2}+\mathbf{3}+\mathbf{4}+\mathbf{5}}_B],N_4:=\mathbb{I}-(N_1+N_2+N_3)\}$. Here, $P[(\ket{i},\ket{j})_\#]:=(\ket{i}\bra{i}+\ket{j}\bra{j})_\#$, and $\#$ denotes the party. When $N_1$ clicks, the given state must be $\ket{\psi_3}$. Similarly, for the click $N_2$, the state is $\ket{\psi_4}$, and for $N_3$ it is $\ket{\psi_5}$. Whenever $N_4$ clicks the given state can be either $\ket{\psi_1}$ or $\ket{\psi_2}$. However, in that case, Alice can perform a measurement to distinguish between these two \cite{lid3,walgate}. This concludes the local distinguishability protocol for the set $\mathcal{G}_1$.
\end{proof}


In the following, we will demonstrate a protocol to activate the nonlocality without entanglement from the set $\mathcal{G}_{1}$. 
\begin{theorem}
The locally distinguishable set $\mathcal{G}_{1}$ can be transformed deterministically to a locally indistinguishable set via local OPM.
\end{theorem}
\begin{proof}
Consider that Bob performs a local OPM on the subsystem $B$: $K_B\equiv\{K_1:=P[(\ket{\mathbf{0}},\ket{\mathbf{1}},\ket{\mathbf{2}})_B], K_2:=P[(\ket{\mathbf{3}},\ket{\mathbf{4}},\ket{\mathbf{5}})_B]\}$. If $K_1$ clicks they end up in one of
\begin{subequations}
\begin{align*}
\left\{\!\begin{aligned}
\ket{0}_{A}\ket{\mathbf{0}-\mathbf{1}}_{B},~\ket{2}_{A}\ket{\mathbf{1}-\mathbf{2}}_{B},\\\ket{1-2}_{A}\ket{\mathbf{0}}_{B},~
\ket{0-1}_{A}\ket{\mathbf{2}}_{B},\\~\ket{0+1+2}_{A}\ket{\mathbf{0}+\mathbf{1}+\mathbf{2}}_{B}.~~~
\end{aligned}\right\}	
\end{align*}
\end{subequations}
On the other hand, if Bob gets $K_2$, they are then left with one of the following five states: 
\begin{subequations}
\begin{align*}
\left\{\!\begin{aligned}
\ket{0}_{A}\ket{\mathbf{4}-\mathbf{5}}_{B},~\ket{2}_{A}\ket{\mathbf{5}-\mathbf{3}}_{B},\\\ket{1-2}_{A}\ket{\mathbf{4}}_{B},~
\ket{0-1}_{A}\ket{\mathbf{3}}_{B},\\~\ket{0+1+2}_{A}\ket{\mathbf{3}+\mathbf{4}+\mathbf{5}}_{B}.~~~
\end{aligned}\right\}	
\end{align*}
\end{subequations}
It is clear that the five updated states when $K_1$ clicks form the celebrated unextendible product basis (UPB) \cite{upb,cmp} in $\mathbb{C}^3 \otimes \mathbb{C}^3$. The states in case of $K_2$ outcome also form the same UPB where $\{\ket{\mathbf{3}}_B,\ket{\mathbf{4}}_B,\ket{\mathbf{5}}_B\}$ span Bob's three dimensional Hilbert space. It has been well established that this orthogonal set of product states is locally indistinguishable \cite{bennps,cmp}. 
\end{proof}
This is certainly an example of genuine activation of bipartite quantum nonlocality without entanglement. However, orthogonal sets of bipartite product states that show activable nonlocality have already been reported \cite{liz}. But in some of the protocols, mentioned there, different outcome of a single local OPM provides different dimensional sets of nonlocal product states. Moreover, the dimension requirement to activate such nonlocality in \cite{liz} is minimum $7\times8$ for  bipartite systems, while our elegant example shows that such an activation is possible even with lower dimensional quantum systems.

The question of genuine activation of multipartite quantum nonlocality from orthogonal sets of product states also has not been explored yet. In the following, we delve into this question and answer in affirmative with an explicit example. Consider the orthogonal set of tripartite product states $\mathcal{G}_2\equiv\{\ket{\phi_i}_{ABC}\}_{i=1}^4(\subset \mathbb{C}^2\otimes \mathbb{C}^2\otimes\mathbb{C}^4)$ where,
\begin{subequations}
\begin{eqnarray}
\ket{\phi_1}_{ABC}&=&\ket{0}_{A}\ket{0-1}_B\ket{\mathbf{1}+\mathbf{2}}_{C}\\
\ket{\phi_2}_{ABC}&=&\ket{0-1}_{A}\ket{1}_B\ket{\mathbf{0}+\mathbf{3}}_{C}\\
\ket{\phi_3}_{ABC}&=&\ket{1}_{A}\ket{0}_B\ket{\mathbf{0}-\mathbf{1}+\mathbf{2}-\mathbf{3}}_{C}\\
\ket{\phi_4}_{ABC}&=&\ket{0+1}_{A}\ket{0+1}_B\ket{\mathbf{0}+\mathbf{1}+\mathbf{2}+\mathbf{3}}_{C},
\end{eqnarray}
\end{subequations}
\begin{proposition}
The set $\mathcal{G}_{2}$ is free from local redundancy and discriminable under LOCC.
\end{proposition}
\begin{proof}
It is straightforward to show that the set $\mathcal{G}_2$ is free form local redundancy. Here, Charlie's system can be considered as two composite qubits. Let us denote those subsystems by $c_1$ and $c_2$: $\ket{\mathbf{0}}_C:=\ket{00}_{c_1c_2},\ket{\mathbf{1}}_C:=\ket{01}_{c_1c_2},\ket{\mathbf{2}}_C:=\ket{10}_{c_1c_2},\ket{\mathbf{3}}:=\ket{11}_{c_1c_2}$. Consider the states $\ket{\phi_1}$ and $\ket{\phi_2}$. Note that discarding the subsystem $c_i$ we will have $\rho^k_{c_j}:=Tr_{c_i}\ket{\phi_k}\bra{\phi_k}$, for $i,j,k\in\{1,2\}$. It is quite evident that $\rho^1_{c_j}$ and $\rho^2_{c_j}$ are non-orthogonal for $j=1,2$.  

We will now show that the set $\mathcal{G}_2$ is locally distinguishable. The distinguishability protocol is as follows. First, Charlie performs a measurement $K_C\equiv\{K_1:=P[\ket{\mathbf{0}+\mathbf{3}}_C],K_2:=P[\ket{\mathbf{0}-\mathbf{3}}_C],K_3:=P[\ket{\mathbf{1}+\mathbf{2}}_C],K_4:=P[\ket{\mathbf{1}-\mathbf{2}}_C]\}$. If $K_1$ clicks the given state must be one of $\{\ket{\phi_2}_{ABC}, \ket{\phi_4}_{ABC}\}$ which are perfectly locally distinguishable \cite{lid3,walgate}. If $K_2$ clicks, the given state must be $\ket{\phi_2}_{ABC}$. When $K_3$ clicks, the given state is one of $\{\ket{\phi_1}_{ABC}, \ket{\phi_4}_{ABC}\}$ which can always be perfectly distinguished via LOCC. For the click $K_4$, the given state is certainly $\ket{\phi_3}_{ABC}$.  
\end{proof}

Now, we are in a position to show that the set $\mathcal{G}_2$ can be transformed, with certainty, to a set of orthogonal states which are impossible to distinguish locally. 
\begin{theorem}
The set $\mathcal{G}_2$ can be converted to a set of tripartite locally indistinguishable product states, \textit{a.k.a}, the Shifts UPB \cite{upb} using local OPM.  
\end{theorem}
\begin{proof}
Let us consider that Charlie performs a local OPM on the subsystem $C$: $R_B\equiv\{R_1:=P[(\ket{\mathbf{0}},\ket{\mathbf{1}})_C], R_2:=P[(\ket{\mathbf{2}},\ket{\mathbf{3}})_C]\}$. If $R_1$ clicks they end up in one of 
\begin{eqnarray}\nonumber
\{\ket{0}_{A}\ket{0-1}_B\ket{\mathbf{1}}_C,\ket{0-1}_{A}\ket{1}_B\ket{\mathbf{0}}_C,\\\nonumber\ket{1}_{A}\ket{0}_B\ket{\mathbf{0-1}}_C,\ket{0+1}_{A}\ket{0+1}_B\ket{\mathbf{0+1}}_C\}.
\end{eqnarray}
On the other hand, if $R_2$ clicks, they will be left with one of 
\begin{eqnarray}\nonumber
\{\ket{0}_{A}\ket{0-1}_B\ket{\mathbf{2}}_C,\ket{0-1}_{A}\ket{1}_B\ket{\mathbf{3}}_C,\\\nonumber\ket{1}_{A}\ket{0}_B\ket{\mathbf{2-3}}_C,\ket{0+1}_{A}\ket{0+1}_B\ket{\mathbf{2+3}}_C\}.
\end{eqnarray}
It is evident that both the above sets are basically equivalent to the Shifts UPB in $\mathbb{C}^2\otimes\mathbb{C}^2\otimes\mathbb{C}^2$ \cite{upb,cmp,bennps}. Furthermore, states belonging to Shifts UPB are known to be perfectly indistinguishable  via LOCC \cite{cmp}. This completes our proof.
\end{proof}
Note that, whenever any two parties of the above Shifts UPB come together, the set turns out to be locally distinguishable. Any orthogonal multipartite set of product states that can be thought of as a bipartition of $\mathbf{C}^2\otimes \mathbb{C}^d$ with $d\ge2$, can always be shown as locally distinguishable \cite{cmp}. Therefore, it is quite clear that this set must not be locally irreducible in any bipartition.
\section{Genuine Activation of Strongest Possible Quantum Nonlocality Without Entanglement}
After demonstrating the activation of genuine nonlocality without entanglement in the multipartite scenario, the pertinent question is whether or not there exists any orthogonal set which can show genuine activation of the strongest possible form of quantum nonlocality without entanglement as demonstrated in \cite{halderprl}.  Note that, strong quantum nonlocality without entanglement \cite{halderprl} can not be obtained in mere three-qubit systems. The minimum dimension required to show such phenomena is at least three qutrit. In the following, we provide a set which answers this question in affirmation. Consider the orthogonal set $\mathcal{G}_3$ that contains the following of $27$ tripartite product states\footnote{For the sake of better readability here we drop the party notation ($A,B,C$) in the subscripts} $\ket{\xi_i^{\pm}},~i\in\{1,\cdots,4,6,\cdots,9,11,\cdots,14\}$ and $\ket{\xi_j},~j\in\{5,10,15\}$ in $\mathbb{C}^6\otimes \mathbb{C}^6\otimes\mathbb{C}^6$.
\begin{subequations}\label{str}
\begin{eqnarray}
\ket{\xi_1^{\pm}}&=&\ket{\mathbf{0}-\mathbf{4}}\ket{\mathbf{1}-\mathbf{5}}\ket{\mathbf{0}\pm\mathbf{1}+\mathbf{4}\pm\mathbf{5}}\\
\ket{\xi_2^{\pm}}&=&\ket{\mathbf{0}-\mathbf{4}}\ket{\mathbf{2}-\mathbf{3}}\ket{\mathbf{0}\pm\mathbf{2}+\mathbf{4}\pm\mathbf{3}}\\
\ket{\xi_3^{\pm}}&=&\ket{\mathbf{1}-\mathbf{5}}\ket{\mathbf{2}-\mathbf{3}}\ket{\mathbf{0}\pm\mathbf{1}+\mathbf{4}\pm\mathbf{5}}\\
\ket{\xi_4^{\pm}}&=&\ket{\mathbf{2}-\mathbf{3}}\ket{\mathbf{1}-\mathbf{5}}\ket{\mathbf{0}\pm\mathbf{2}+\mathbf{4}\pm\mathbf{3}}\\
\ket{\xi_5}&=&\ket{\mathbf{0}-\mathbf{4}}\ket{\mathbf{0}-\mathbf{4}}\ket{\mathbf{0}-\mathbf{4}}\\
\ket{\xi_6^{\pm}}&=&\ket{\mathbf{1}-\mathbf{5}}\ket{\mathbf{0}\pm\mathbf{1}+\mathbf{4}\pm\mathbf{5}}\ket{\mathbf{0}-\mathbf{4}}\\
\ket{\xi_7^{\pm}}&=&\ket{\mathbf{2}-\mathbf{3}}\ket{\mathbf{0}\pm\mathbf{2}+\mathbf{4}\pm\mathbf{3}}\ket{\mathbf{0}-\mathbf{4}}\\
\ket{\xi_8^{\pm}}&=&\ket{\mathbf{2}-\mathbf{3}}\ket{\mathbf{0}\pm\mathbf{1}+\mathbf{4}\pm\mathbf{5}}\ket{\mathbf{1}-\mathbf{5}}\\
\ket{\xi_9^{\pm}}&=&\ket{\mathbf{1}-\mathbf{5}}\ket{\mathbf{0}\pm\mathbf{2}+\mathbf{4}\pm\mathbf{3}}\ket{\mathbf{2}-\mathbf{3}}\\
\ket{\xi_{10}}&=&\ket{\mathbf{1}-\mathbf{5}}\ket{\mathbf{1}-\mathbf{5}}\ket{\mathbf{1}-\mathbf{5}}\\
\ket{\xi_{11}^{\pm}}&=&\ket{\mathbf{0}\pm\mathbf{1}+\mathbf{4}\pm\mathbf{5}}\ket{\mathbf{0}-\mathbf{4}}\ket{\mathbf{1}-\mathbf{5}}\\
\ket{\xi_{12}^{\pm}}&=&\ket{\mathbf{0}\pm\mathbf{2}+\mathbf{4}\pm\mathbf{3}}\ket{\mathbf{0}-\mathbf{4}}\ket{\mathbf{2}-\mathbf{3}}\\
\ket{\xi_{13}^{\pm}}&=&\ket{\mathbf{0}\pm\mathbf{1}+\mathbf{4}\pm\mathbf{5}}\ket{\mathbf{1}-\mathbf{5}}\ket{\mathbf{2}-\mathbf{3}}\\
\ket{\xi_{14}^{\pm}}&=&\ket{\mathbf{0}\pm\mathbf{2}+\mathbf{4}\pm\mathbf{3}}\ket{\mathbf{2}-\mathbf{3}}\ket{\mathbf{1}-\mathbf{5}}\\
\ket{\xi_{15}}&=&\ket{\mathbf{2}-\mathbf{3}}\ket{\mathbf{2}-\mathbf{3}}\ket{\mathbf{2}-\mathbf{3}},
\end{eqnarray}
\end{subequations}  
\begin{proposition}\label{proposition2}
The set $\mathcal{G}_{3}$ is not locally redundant and is distinguishable under LOCC, even when all the parties are separated. 
\end{proposition}
\begin{proof}
We first provide a brief outline of the proof that the above set of states are free from local redundancy. The detailed proof is given in the appendix. Note that, the quantum system possessed by each individuals can only be composed of a qubit and qutrit subsystem. Therefore, for each player we can write: $\ket{\mathbf{0}}:=\ket{00},\ket{\mathbf{1}}:=\ket{01},\ket{\mathbf{2}}:=\ket{02},\ket{\mathbf{3}}:=\ket{10},\ket{\mathbf{4}}:=\ket{11},\ket{\mathbf{5}}:=\ket{12}$. First consider the players discard their subsystems in such a way that they ultimately get the dimension of the whole tripartite system below $27$. This is possible when more than one player discard their qutrits (that is, $\mathbb{C}^{2}\otimes\mathbb{C}^{2}\otimes\mathbb{C}^{2}$, or, $\mathbb{C}^{2}\otimes\mathbb{C}^{2}\otimes\mathbb{C}^{6}$, or, $\mathbb{C}^{2}\otimes\mathbb{C}^{3}\otimes\mathbb{C}^{3}$ etc.). In this case, it is clear that all the states in (\ref{str}) will not retain their orthogonality. Other possible cases of discarding the sub-parts can be branched as follows: any one player discards their qutrit, any one player discards their qubit and more than one player discard their qubits. Though cumbersome but it is quite straightforward to show that in all these cases the reduced states' set will not be orthogonal anymore. Therefore, we can conclude that the set $\mathcal{G}_3$ does not have local redundancy.  

Now, we move to the proof that the set $\mathcal{G}_3$ can be distinguished with the help of LOCC alone. Due to the symmetries present in the set $\mathcal{G}_3$, each player may need to perform any of the following three measurements at different steps of the protocol.
{\small
\begin{eqnarray}\nonumber
\mathcal{M}_{1}&&\equiv\{P_1:=P[\ket{\mathbf{0-4}}],P_2:={P}[\ket{\mathbf{1-5}}], P_3:={P}[\ket{\mathbf{2-3}}],\\\nonumber&&P_4:=\mathbb{I}-({P}[\ket{\mathbf{0-4}}]+{P}[\ket{\mathbf{1-5}}]+ {P}[\ket{\mathbf{2-3}}])\}\\\nonumber
\mathcal{M}_{2}&&\equiv\{Q_1:=P[\ket{\mathbf{0+1+4+5}}],Q_2:=P[\ket{\mathbf{0-1+4-5}}],\\\nonumber&&Q_3:=\mathbb{I}-(P[\ket{\mathbf{0+1+4+5}}]+[\ket{\mathbf{0-1+4-5}}])\}\\\nonumber
\mathcal{M}_{3}&&\equiv\{R_1:=P[\ket{\mathbf{0+2+4+3}}],R_2:=P[\ket{\mathbf{0-2+4-3}}],\\\nonumber&&R_3:=\mathbb{I}-(P[\ket{\mathbf{0+2+4+3}}]+{P}[\ket{\mathbf{0-2+4-3}}])\}.
\end{eqnarray}
}
The detailed protocol is pictorially described in the Appendix.
\end{proof}
\begin{theorem}
The set $\mathcal{G}_3$ can be deterministically transformed via local OPMs to a orthogonal set of tripartite product states which are locally irreducible even if all but one player come together. 
\end{theorem}
\begin{proof}
Suppose, each player performs a specific orthogonality preserving local measurement: $\mathcal{K}\equiv\{K_1:=P[\ket{\mathbf{0}},\ket{\mathbf{1}},\ket{\mathbf{2}}],K_2:=P[\ket{\mathbf{3}},\ket{\mathbf{4}},\ket{\mathbf{5}}]\}$. Here, the notation we follow is as follows: $K_j^i$ is the $j$th projector ($K_j$) that clicks when $i$th player performs the measurement $\mathcal{K}$. Therefore, after the measurement a total of eight possibilities can occur as each player can get any one of two possible outcomes $K_1^i$ or $K_2^i$. In each of these eight cases it is straightforward to see that the updated set of $27$ states will be of the following generic form. 
\begin{align}\label{str1}
\left\{\!\begin{aligned}
\tri{p}{q}{\eta_\pm},~\tri{q}{\eta_\pm}{p},~\tri{\eta_\pm}{p}{q},\\
\tri{p}{r}{\kappa_\pm},~\tri{r}{\kappa_\pm}{p},~\tri{\kappa_\pm}{p}{r},\\
\tri{q}{r}{\eta_\pm},~\tri{r}{\eta_\pm}{q},~\tri{\eta_\pm}{q}{r},\\
\tri{r}{q}{\kappa_\pm},~\tri{q}{\kappa_\pm}{r},~\tri{\kappa_\pm}{r}{q},\\
\tri{p}{p}{p},~\tri{q}{q}{q},~\tri{r}{r}{r}~~~
\end{aligned}\right\}	
\end{align}
Here, $\ket{\eta_\pm}:=(\ket{p}\pm\ket{q})/\sqrt{2}$ and $\ket{\kappa_\pm}:=(\ket{p}\pm\ket{r})/\sqrt{2}$. In each of the eight outcomes, for all $27$ states $p,q$ and $r$ will have some specific values from $p\in\{\mathbf{0,4}\},q\in\{\mathbf{1,5}\}$ and $r\in\{\mathbf{2,3}\}$. For example, if for all the players, the outcomes are $K_1$ throughout then the reduced set of states will be of the above form with $p=\{\mathbf{0}\},q=\{\mathbf{1}\}$ and $r=\{\mathbf{2}\}$.

Note that the above set of states is basically the orthogonal set that manifests strong quantum nonlocality without entanglement \cite{halderprl}.
\end{proof}
At this end, one may be further curious to activate such a strongest possible genuine quantum nonlocality without entanglement by performing a measurement on the possession of a single party. This has vivid importance in the framework of data hiding and secret sharing between all but one untrusted parties. Precisely speaking, in such a scenario the particular trusted agent (personified as Charlie) has full authority to judge how trustworthy are the other parties and depending upon that he may compel others to meet him in person to decode a hidden secret. As an example consider the following set $\mathcal{G}_{4}$ of  $27$ orthogonal product states $\ket{\zeta_i^{\pm}},~i\in\{1,\cdots,4,6,\cdots,9,11,\cdots,14\}$ and $\ket{\zeta_j},~j\in\{5,10,15\}$ in  $\mathbb{C}^{3^{\otimes2}}\otimes\mathbb{C}^{6}$,
\begin{subequations}\label{stra}
\begin{eqnarray}
\ket{\zeta_1^{\pm}}&=&\ket{0}\ket{1}\ket{\mathbf{0}\pm\mathbf{1}+\mathbf{4}\pm\mathbf{5}}\\
\ket{\zeta_2^{\pm}}&=&\ket{0}\ket{2}\ket{\mathbf{0}\pm\mathbf{2}+\mathbf{4}\pm\mathbf{3}}\\
\ket{\zeta_3^{\pm}}&=&\ket{{1}}\ket{{2}}\ket{\mathbf{0}\pm\mathbf{1}+\mathbf{4}\pm\mathbf{5}}\\
\ket{\zeta_4^{\pm}}&=&\ket{{2}}\ket{{1}}\ket{\mathbf{0}\pm\mathbf{2}+\mathbf{4}\pm\mathbf{3}}\\
\ket{\zeta_5}&=&\ket{{0}}\ket{{0}}\ket{\mathbf{0}-\mathbf{4}}\\
\ket{\zeta_6^{\pm}}&=&\ket{{1}}\ket{{0}\pm{1}}\ket{\mathbf{0}-\mathbf{4}}\\
\ket{\zeta_7^{\pm}}&=&\ket{{2}}\ket{{0}\pm{2}}\ket{\mathbf{0}-\mathbf{4}}\\
\ket{\zeta_8^{\pm}}&=&\ket{{2}}\ket{{0}\pm{1}}\ket{\mathbf{1}-\mathbf{5}}\\
\ket{\zeta_9^{\pm}}&=&\ket{{1}}\ket{{0}\pm{2}}\ket{\mathbf{2}-\mathbf{3}}\\
\ket{\zeta_{10}}&=&\ket{{1}}\ket{{1}}\ket{\mathbf{1}-\mathbf{5}}\\
\ket{\zeta_{11}^{\pm}}&=&\ket{{0}\pm{1}}\ket{{0}}\ket{\mathbf{1}-\mathbf{5}}\\
\ket{\zeta_{12}^{\pm}}&=&\ket{{0}\pm{2}}\ket{{0}}\ket{\mathbf{2}-\mathbf{3}}\\
\ket{\zeta_{13}^{\pm}}&=&\ket{{0}\pm{1}}\ket{{1}}\ket{\mathbf{2}-\mathbf{3}}\\
\ket{\zeta_{14}^{\pm}}&=&\ket{{0}\pm{2}}\ket{{2}}\ket{\mathbf{1}-\mathbf{5}}\\
\ket{\zeta_{15}}&=&\ket{{2}}\ket{{2}}\ket{\mathbf{2}-\mathbf{3}},
\end{eqnarray}
\end{subequations}


\begin{proposition}\label{propositiona}
The set $\mathcal{G}_{4}$ is distinguishable under LOCC and free from local redundancy. 
\end{proposition}
\begin{proof}
The proof that the set $\mathcal{G}_4$ does not have redundancy is quite straightforward. One may consider that the subsystem of Charlie ($\mathbb{C}^6$) consists of a qubit and qutrit. Now, if we discard any of the qubit or qutrit parts, not all pairs that remain would be orthogonal. The proof is quite evident from the proof of proposition \ref{proposition2}. 

We will now move to describe a local discrimination protocol of the set $\mathcal{G}_4$. We will provide here a brief outline of the protocol.

Charlie first performs a measurement
{\small
\begin{eqnarray}\nonumber
\mathcal{M}^C_{1}&&\equiv\{P_1:=P[\ket{\mathbf{0-4}}_C],P_2:={P}[\ket{\mathbf{1-5}}_C], P_3:={P}[\ket{\mathbf{2-3}}_C],\\\nonumber&&P_4:=\mathbb{I}-({P}[\ket{\mathbf{0-4}}_C]+{P}[\ket{\mathbf{1-5}}_C]+ {P}[\ket{\mathbf{2-3}}_C])\}
\end{eqnarray}}
Now, depending upon different outcomes, Bob and Charlie will perform some suitable measurements at their local sites to distinguish the set. A step by step detailed analysis is provided in the Appendix.
\end{proof}

\begin{theorem}
The set $\mathcal{G}_4$ can be deterministically transformed, via a single local OPM at Charlie's site, to an orthogonal set of tripartite product states which is locally irreducible in every bipartition. 
\end{theorem}
\begin{proof}
Consider that Charlie performs a local OPM, 
$\mathcal{K}^C\equiv\{K_1^C:=P[\ket{\mathbf{0}},\ket{\mathbf{1}},\ket{\mathbf{2}}_C],K_2^C:=P[\ket{\mathbf{3}},\ket{\mathbf{4}},\ket{\mathbf{5}}_C]\}$
For different outcomes of $\mathcal{K}^C$, the post measurement state will turn out to be any of the following set:
\begin{align}\label{str2}
\left\{\!\begin{aligned}
\tri{0}{1}{{\tilde{\nu}_\pm}},~\tri{1}{\nu_\pm}{p},~\tri{\nu_\pm}{o}{q},\\
\tri{0}{2}{\tilde{\tau}_\pm},~\tri{2}{\tau_\pm}{p},~\tri{\tau_\pm}{o}{r},\\
\tri{1}{2}{\tilde{\nu}_\pm},~\tri{2}{\nu_\pm}{q},~\tri{\nu_\pm}{1}{r},\\
\tri{2}{1}{\tilde{\tau}_\pm},~\tri{1}{\tau_\pm}{r},~\tri{\tau_\pm}{2}{q},\\
\tri{0}{0}{p},~\tri{1}{1}{q},~\tri{2}{2}{r}~~~
\end{aligned}\right\}	
\end{align}
where, Here, $\ket{\nu_\pm}:=(\ket{0}\pm\ket{1})/\sqrt{2}$, $\ket{\tau_\pm}:=(\ket{0}\pm\ket{2})/\sqrt{2}$, $\ket{\tilde{\nu}_\pm}:=(\ket{p}\pm\ket{q})/\sqrt{2}$ and $\ket{\tilde{\tau}_\pm}:=(\ket{p}\pm\ket{r})/\sqrt{2}$. $p,q$ and $r$ can have any value $\{\mathbf{0,4}\}$, $\{\mathbf{1,5}\}$ and $\{\mathbf{2,3}\}$ respectively.
Now, when $K_1^C$ clicks, the post measurement state can be any of the set (\ref{str2}) with $(p,q,r)=(\mathbf{0,1,2})$. Otherwise, if $K_2^C$ clicks, they are left with any of set (\ref{str2}) where $(p,q,r)=(\mathbf{4,5,3})$. It is evident that the set (\ref{str2}) shows strong quantum nonlocality without entanglement \cite{halderprl,gnps}. This completes our proof.

\end{proof}

\section{Discussion}
In summary, we have studied the genuine activation of nonlocality from several sets of local states. However, the phrases "local" and "nonlocal" have been used from the state discrimination perspective. More precisely, we have dealt with two different sets of locally distinguishable multipartite product states in $\mathbb{C}^{3}\otimes\mathbb{C}^{6}$ and $\mathbb{C}^{2^{\otimes2}}\otimes\mathbb{C}^{4}$, which can be transformed to the set of locally indistinguishable states in $\mathbb{C}^{3^{\otimes2}}$ and $\mathbb{C}^{2^{\otimes3}}$ respectively, by choosing an appropriate measurement in possession of one of the parties. Furthermore, we have considered a stronger notion of state discrimination problem, namely the orthogonality preserving reducibility and have shown to activate such a notion from a set of multipartite locally distinguishable product states of dimension $\mathbb{C}^{6^{\otimes3}}$. It is observed that under LOCC, the set can be transformed deterministically to a set of states in $\mathbb{C}^{3^{\otimes3}}$, which is even irreducible in all possible bipartitions. Further, we have moved to a stricter notion of such activation where the transformation is possible to achieve by a single agent only. Such an example is demonstrated to transform a locally distinguishable set of states of $\mathbb{C}^{3^{\otimes2}}\otimes\mathbb{C}^{6}$ to a strongly irreducible one in $\mathbb{C}^{3^{\otimes3}}$. The elegance of state construction and the transformation protocol makes it trivial to extend in any arbitrary higher dimensional set of product states exhibiting nonlocality in terms of local discrimination and orthogonality preserving elimination. Besides its foundational interest to understand the topology of the state spaces of composite quantum systems, our work deserves significant importance from the practical perspective. It has mimicked an interesting framework of secured data hiding between several parties, where the distributor is flexible to update the distinguishibility of the secured data hidden in the correlation of the given states. 

\section{Acknowledgement}
We would like to acknowledge stimulating discussions with Guruprasad Kar, Manik Banik, Saronath Halder and Ramij Rahaman. Tamal Guha would like to acknowledge his academic visit at Indian Statistical Institute, Kolkata, during November-December of 2021.


\bibliography{SR}

\onecolumngrid 
\appendix
\onecolumngrid 
\section*{Proof of Proposition \ref{proposition2}}
\begin{proof}
{We will first prove that the above set of states are free from local redundancy through a detailed case-wise analysis:}

{\textbf{Case-I: \textit{More than one} players discard their qutrits:} Observe that, such operations lead to $27$ tripartite states in $\mathbb{C}^{2}\otimes\mathbb{C}^{2}\otimes\mathbb{C}^{2}$ (when all the qutrits have been traced out) or in $\mathbb{C}^{2}\otimes\mathbb{C}^{2}\otimes\mathbb{C}^{6}$ or its party permutations (when two of the qutrits have been traced out). Such reduced sets can not be orthogonal in any way.}

{\textbf{Case-II: \textit{Any one} player discards their qutrit:} Without loss of generality, we can analyze the case when Charlie has traced out the qutrit system. Same argument also works for other two players.

Note that under such an operation (discarding the qutrit) only possibility of non-orthogonality arise from any of the twelve pairs of $\{\ket{\xi_{i}^{\pm}}\}$.

Again, from the symmetric structure of these states, the non-orthogonality of all those $12$ pairs can be sufficiently guaranteed by checking only $\{\ket{\xi_{1}^{\pm}}\}$    and $\{\ket{\xi_{6}^{\pm}}\}$. Checking the first one is sufficient for $i\in\{1,2,3,4\}$ pairs and the second one is sufficient for $i\in\{6,7,8,9,11,12,13,14\}$ pairs.
From now on, we will use the party annotation for clarity. We will denote each party's qubit system with $1$ in subscript and qutrit with $2$ in subscript. For example, $\ket{\xi_{1}^{\pm}}_{{a_1a_2b_1b_2c_1c_2}}$ implies that Alice, Bob and Charlies' qubits are denoted by the subsystems $a_1,b_1$ and $c_1$, whereas their qutrit subsystems are denoted by $a_2,b_2$ and $c_2$. 

Now, tracing out Charlie's qutrit from $\ket{\xi_{1}^{\pm}}_{{a_1a_2b_1b_2c_1c_2}}$, we will obtain, $\rho_{{a_1a_2b_1b_2c_1}}^{1\pm}=\text{Tr}_{c_2}(\ket{\xi_{1}^{\pm}}\bra{\xi_{1}^{\pm}}_{{a_1a_2b_1b_2c_1c_2}})=\ket{\mathbf{0-4}}\bra{\mathbf{0-4}}_{a_1a_2}\otimes\ket{\mathbf{1-5}}\bra{\mathbf{1-5}}_{b_1b_2}\otimes(\frac{1}{2}\ket{\alpha_\pm}\bra{\alpha_\pm}+\frac{1}{2}\times\frac{\mathbb{I}}{2})_{c_1}$, where $\ket{\alpha_\pm}=\frac{\ket{0}\pm\ket{1}}{\sqrt{2}}$ and the superscript of $\rho_{{a_1a_2b_1b_2c_1}}^{1\pm}$ implies that it is a reduced state of $\ket{\xi_{1}^{\pm}}$. It trivially follows that the states $\rho_{{a_1a_2b_1b_2c_1}}^{1\pm}$ are not orthogonal to each other.}

{\textbf{Case-III: \textit{Any one} player discards their qubit:} It is sufficient to analyze the case $\mathbb{C}^{6}\otimes\mathbb{C}^{6}\otimes\mathbb{C}^{3}$. This does not lose any generality.

Consider the states $\ket{\xi_{1}^{\pm}}_{{a_1a_2b_1b_2c_1c_2}}$. Discarding Charlie's qubit from this state, we get $\rho_{{a_1a_2b_1b_2c_2}}^{1\pm}={Tr}_{c_1}(\ket{\xi_{1}^{\pm}}\bra{\xi_{1}^{\pm}}_{{a_1a_2b_1b_2c_1c_2}})=\ket{\mathbf{0-4}}\bra{\mathbf{0-4}}_{a_1a_2}\otimes\ket{\mathbf{1-5}}\bra{\mathbf{1-5}}_{b_1b_2}\otimes(\frac{1}{2}\ket{\alpha_\pm}\bra{\alpha_\pm}+\frac{1}{2}\ket{\beta_{\pm}}\bra{\beta_{\pm}})_{c_2}$, where $\ket{\alpha_\pm}=\frac{\ket{0}\pm\ket{1}}{\sqrt{2}}$ and $\ket{\beta_{\pm}}=\frac{\ket{1}\pm\ket{2}}{\sqrt{2}}$. Note that, $\ket{\alpha_\pm}$ is not orthogonal to $\ket{\beta_{\pm}}$, it is, thus, straightforward to assure the non-orthogonality of $\rho_{{a_1a_2b_1b_2c_2}}^{1\pm}$.}

{\textbf{Case-IV: \textit{More than one} players discard their qubits:} The non-orthogonality depicted in Case-III, guarantees that tracing out more systems from $\rho_{{a_1a_2b_1b_2c_2}}^{1\pm}$ can not make them orthogonal.}\par

\begin{figure}[h!]
	\centering
	\includegraphics[width=5in]{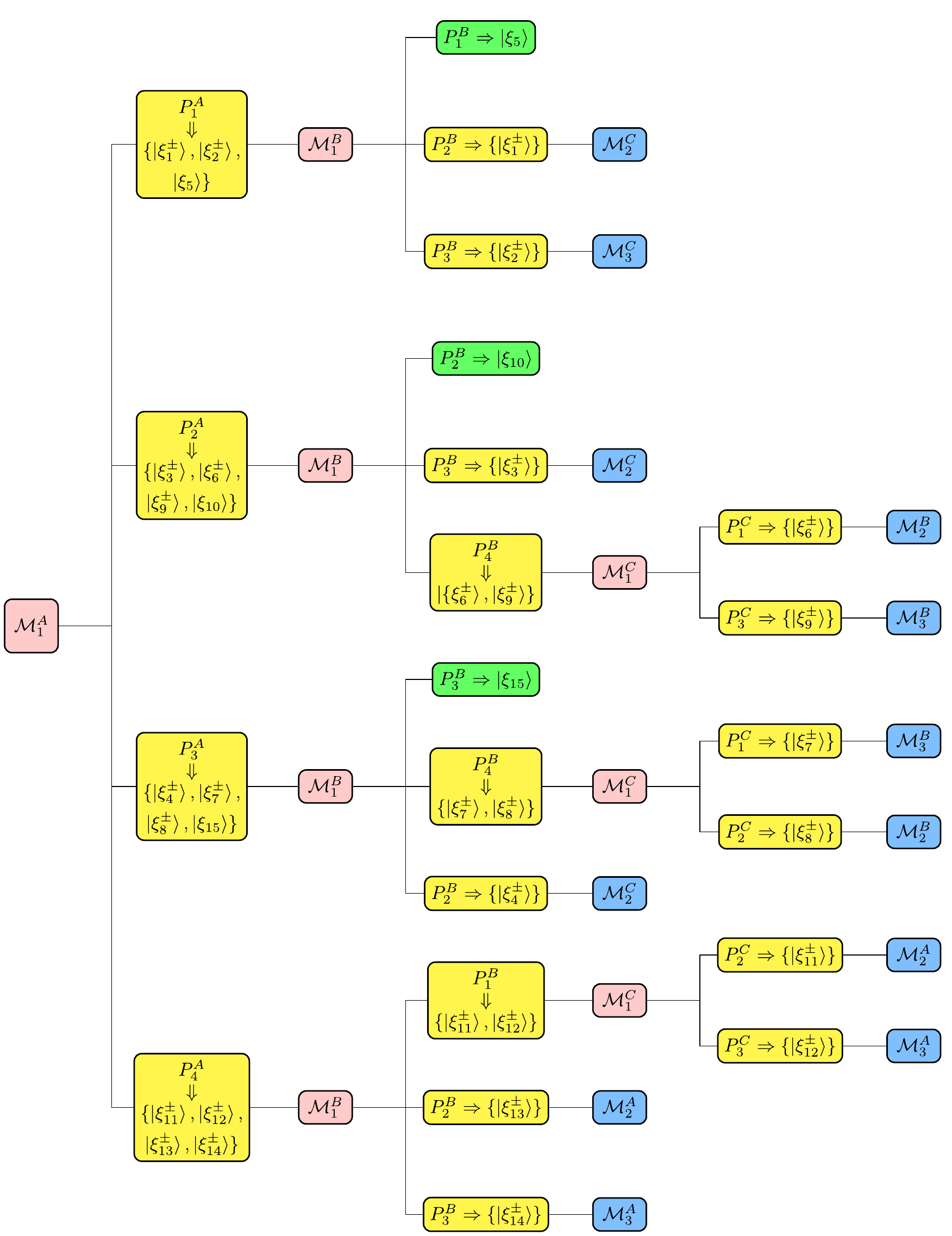}
	\caption{\textbf{Local discrimination protocol for the set $\mathcal{G}_3$.} The green boxes stand for the case where the state has been distinguished by clicking of a particular projector and the blue boxes refer to the ultimate measurement which has to be performed on a set of two states to perfectly distinguish them.}
	\label{fig1}
\end{figure}

\textbf{Case-V: \textit{Any one} player discards their qubits and any other discards their qutrit:}
This case is also a subset of Case-III. It is evident that discarding any qutrits from Alice's or Bob's system does not make the states $\rho_{{a_1a_2b_1b_2c_2}}^{1\pm}$ orthogonal.
{This completes our proof that $\mathcal{G}_3$ does not have local redundancy.}

We now move to the proof of local distinguishability of the set $\mathcal{G}_3$.
{We first define some important measurements which will be required to perform by the players at different steps of the protocol:}

\begin{eqnarray}\nonumber
\mathcal{M}_{1}&&\equiv\{P_1:=P[\ket{\mathbf{0-4}}],P_2:={P}[\ket{\mathbf{1-5}}], P_3:={P}[\ket{\mathbf{2-3}}],P_4:=\mathbb{I}-({P}[\ket{\mathbf{0-4}}]+{P}[\ket{\mathbf{1-5}}]+ {P}[\ket{\mathbf{2-3}}])\}\\\nonumber
\mathcal{M}_{2}&&\equiv\{Q_1:=P[\ket{\mathbf{0+1+4+5}}],Q_2:=P[\ket{\mathbf{0-1+4-5}}],Q_3:=\mathbb{I}-(P[\ket{\mathbf{0+1+4+5}}]+[\ket{\mathbf{0-1+4-5}}])\}\\\nonumber
\mathcal{M}_{3}&&\equiv\{R_1:=P[\ket{\mathbf{0+2+4+3}}],R_2:=P[\ket{\mathbf{0-2+4-3}}],R_3:=\mathbb{I}-(P[\ket{\mathbf{0+2+4+3}}]+{P}[\ket{\mathbf{0-2+4-3}}])\}.
\end{eqnarray}

{For rest of the proof, we use $\mathcal{M}_{k}^{i}$ to mean that the party-$i\in\{A,B,C\}$ has performed the measurement $k\in\{1,2,3\}$. In a similar fashion, $P_k^i$ implies that $P_k$ projector is clicked when measurement $\mathcal{M}_1^i$ is performed on $i$th party's subsystem. Let us now move to the distinguishability protocol.
First, Alice performs the local measurement $\mathcal{M}_{1}^A$. Consequently, depending upon the given state, any four of the projectors can click. Whatever outcome Alice gets, Bob will perform measurement $\mathcal{M}_1^A$ in the second step and so on. The detailed step by step analysis is pictorially presented in Fig.\ref{fig1}. Note that it is evident that any two pure orthogonal multipartite states can be perfectly locally distinguished always. We thus described our protocol up to the point when the players sole task boils down to locally distinguish two pure orthogonal states. It is also important to mention that the measurements $\mathcal{M}_2$ and $\mathcal{M}_3$ will be mainly required to distinguish two orthogonal pure states at the ultimate steps of the protocol.}
This completes the proof.

\end{proof}
\section*{Proof of Proposition \ref{propositiona}}
\begin{proof}
The proof that the set $\mathcal{G}_4$ is free from local redundancy is quite evident from the proof of Proposition \ref{proposition2}. Thus, here, we will describe only the local distinguishability protocol of $\mathcal{G}_4$.

\textbf{Step-1:}
Charlie performs a measurement: 

\begin{eqnarray}\nonumber
\mathcal{M}^C_{1}&&\equiv\{P_1^C:=P[\ket{\mathbf{0-4}}_C],P_2^C:={P}[\ket{\mathbf{1-5}}_C], P_3^C:={P}[\ket{\mathbf{2-3}}_C],P_4^C:=\mathbb{I}-({P}[\ket{\mathbf{0-4}}_C]+{P}[\ket{\mathbf{1-5}}_C]+ {P}[\ket{\mathbf{2-3}}_C])\}\nonumber
\end{eqnarray}
When $P_1^C$ clicks, the given state must be any of 
$\ket{\zeta_6^{\pm}},\ket{\zeta_7^{\pm}}$ or $\ket{\zeta_5}$. If $P_2^C$ clicks, we can say that the state is one of $\ket{\zeta_8^{\pm}},\ket{\zeta_{11}^{\pm}},\ket{\zeta_{14}^{\pm}}$ or $\ket{\zeta_{10}}$. For the third outcome: click of $P_3^C$ implies that the state can be one of $\ket{\zeta_9^{\pm}},\ket{\zeta_{12}^{\pm}},\ket{\zeta_{13}^{\pm}}$ or $\ket{\zeta_{15}}$. Lastly, if $P_4^C$ clicks, the state can be any of $\ket{\zeta_1^{\pm}},\ket{\zeta_{2}^{\pm}},\ket{\zeta_{3}^{\pm}},\ket{\zeta_{4}^{\pm}}$.  

\textbf{Step-2:} For different outcomes of $\mathcal{M}_1^C$, our protocol splits into different branches. For each outcome of $\mathcal{M}_1^C$, we will describe the next step of our protocol below.

When the outcome $P_1^C$ clicks on measuring $\mathcal{M}_1^C$, three parties are now to discriminate between $\ket{\zeta_6^{\pm}},\ket{\zeta_7^{\pm}}$ and $\ket{\zeta_5}$ states. Now, Alice will perform a measurement 
\begin{eqnarray}\nonumber
\mathcal{M}^A_{1}&&\equiv\{P_1^A:=P[\ket{0}_A],P_2^A:={P}[\ket{{1}}_A], P_3^A:={P}[\ket{{2}}_A]\}
\end{eqnarray}
In case when $P_1^A$ clicks, the players will be sure that the given state is $\ket{\zeta_5}$. Otherwise, when $P_2^A$ or $P_3^A$ clicks, in each case the players are to distinguish between only two states $\ket{\zeta_6^{\pm}}$ or $\ket{\zeta_7^{\pm}}$ respectively and this is always possible \cite{lid3,walgate}.

Consider now that after the measurement $\mathcal{M}_1^C$, $P_2^C$ clicks. Therefore, the task is now to distinguish among $\ket{\zeta_8^{\pm}},\ket{\zeta_{11}^{\pm}},\ket{\zeta_{14}^{\pm}}$ and $\ket{\zeta_{10}}$. At this point, Bob will perform a measurement, 
\begin{eqnarray}\nonumber
\mathcal{M}^B_{1}&&\equiv\{P_1^B:=P[(\ket{0},\ket{1})_B],P_2^B:={P}[\ket{{2}}_B]\}
\end{eqnarray}
Now, if $P_2^B$ clicks the given state must be either of $\ket{\zeta_{14}^{\pm}}$ which can be easy distinguished by a measurement on Alice's side \cite{lid3,walgate}. 
When $P_1^B$ clicks the given state can be any of  $\ket{\zeta_8^{\pm}},\ket{\zeta_{11}^{\pm}}$ and $\ket{\zeta_{10}}$. 
At this point, Alice will perform a measurement:
\begin{eqnarray}\nonumber
\mathcal{M}^A_2&&\equiv\{Q_1^A:=P[(\ket{0},\ket{1})_A],Q_2^A:={P}[\ket{{2}}_A]\}
\end{eqnarray}
If $Q_2^A$ clicks the given state must be either of $\ket{\zeta_8^{\pm}}$ which can again be distinguished perfectly \cite{lid3,walgate}.
When $Q_1^A$ clicks, the players are to distinguish among $\ket{\zeta_{11}^{\pm}}$ and $\ket{\zeta_{10}}$. Now, Bob will perform a measurement: $\mathcal{M}^B_2\equiv\{Q_1^B:=P[\ket{0}_B],Q_2^B:={P}[\ket{{1}}_B,Q_3^B:={P}[\ket{{2}}_B]$. If $Q_2^B$ clicks, then the given state is certainly $\ket{\zeta_{10}}$. Otherwise, if $Q_1^B$ clicks the state is either of $\ket{\zeta_{11}^{\pm}}$ which can be perfectly locally discriminated \cite{lid3,walgate}. 

Now, consider that on measuring $\mathcal{M}_1^C$, $P_3^C$ clicks. Consequently, the players are now to distinguish between $\ket{\zeta_9^{\pm}},\ket{\zeta_{12}^{\pm}},\ket{\zeta_{13}^{\pm}}$ or $\ket{\zeta_{15}}$. Now, Bob will perform a measurement: $\mathcal{M}^B_3\equiv\{T_1^B:=P[(\ket{0},\ket{2})_B],T_2^B:={P}[\ket{{1}}_B]\}$. If $T_2^B$ clicks then the given state must be either of $\ket{\zeta_{13}^{\pm}}$ which can be perfectly locally distinguished \cite{lid3,walgate}. On the other hand, if $T_1^B$ clicks, the state is certainly any of $\ket{\zeta_9^{\pm}},\ket{\zeta_{12}^{\pm}},\ket{\zeta_{15}}$. Now, Alice performs a measurement: $\mathcal{M}^A_2\equiv\{R_1^A:=P[(\ket{0},\ket{2})_A],R_2^A:={P}[\ket{{1}}_A]\}$. If $R_2^A$ clicks, the given state is either of $\ket{\zeta_9^{\pm}}$ which can be perfectly distinguished via LOCC. Otherwise, for the click $R_1^A$ the state can be any of $\ket{\zeta_{12}^{\pm}},\ket{\zeta_{15}}$. At this point, Bob will again perform a measurement, $\mathcal{M}^B_4\equiv\{N_1^B:=P[\ket{0}_B],N_2^B:={P}[\ket{{1}}_B],N_3^B:={P}[\ket{{2}}_B]\}$. If $N_1^B$ clicks the state must be either of $\ket{\zeta_{12}^{\pm}}$ which can be locally distinguishable \cite{lid3,walgate}. Otherwise, if $N_3^B$ clicks the state must be $\ket{\zeta_{15}}$. 

Ultimately, consider the case when $P_4^C$ clicks when $\mathcal{M}_1^C$ is measured. Here, the players are to distinguish between $\ket{\zeta_1^{\pm}},\ket{\zeta_{2}^{\pm}},\ket{\zeta_{3}^{\pm}},\ket{\zeta_{4}^{\pm}}$. Now, Alice will perform a measurement: $\mathcal{M}^A_3\equiv\{T_1^A:=P[\ket{0}_A],T_2^A:={P}[\ket{{1}}_A],T_3^A:={P}[\ket{{2}}_A]\}$. If $T_2^A$ clicks, the given state must be either of $\ket{\zeta_{3}^{\pm}}$ which can be locally distinguished \cite{lid3,walgate}. If $T_3^A$ clicks, the given state is any of $\ket{\zeta_{4}^{\pm}}$ which can again be locally distinguished \cite{lid3,walgate}. Otherwise, if $T_1^A$ clicks then the given state can be any of $\ket{\zeta_1^{\pm}},\ket{\zeta_{2}^{\pm}}$. Bob will now perform a measurement: $\mathcal{M}^B_5\equiv\{L_1^B:=P[\ket{0}_B],L_2^B:={P}[\ket{{1}}_B],L_3^B:={P}[\ket{{2}}_B]\}$. If $L_2^B$ clicks the given state can be either of $\ket{\zeta_1^{\pm}}$ and if $L_3^B$ clicks the given state is any of $\ket{\zeta_2^{\pm}}$. Further, $\ket{\zeta_1^{\pm}}$ (or, $\ket{\zeta_2^{\pm}}$) can be perfectly locally distinguished \cite{lid3,walgate}. This concludes the local distinguishability protocol of $\mathcal{G}_4$.
\end{proof}

\end{document}